\documentclass[letterpaper]{article} % DO NOT CHANGE THIS
\usepackage{aaai23}  % DO NOT CHANGE THIS
\usepackage{times}  % DO NOT CHANGE THIS
\usepackage{helvet}  % DO NOT CHANGE THIS
\usepackage{courier}  % DO NOT CHANGE THIS
\usepackage[hyphens]{url}  % DO NOT CHANGE THIS
\usepackage{graphicx} % DO NOT CHANGE THIS
\usepackage{natbib}  % DO NOT CHANGE THIS AND DO NOT ADD ANY OPTIONS TO IT
\usepackage{caption} % DO NOT CHANGE THIS AND DO NOT ADD ANY OPTIONS TO IT
\DeclareCaptionStyle{ruled}{labelfont=normalfont,labelsep=colon,strut=off} % DO NOT CHANGE THIS
\usepackage{algorithm}
\usepackage{algorithmic}
\usepackage{amsfonts}
\usepackage{amsbsy}
\usepackage{amsmath}
\usepackage{amsthm}
\newtheorem{defi}{Definition}
\newtheorem{theorem}{Theorem}
\newtheorem*{theorem*}{Theorem}
\newtheorem{lemma}{Lemma}

\usepackage{newfloat}
\usepackage{listings}
\floatstyle{ruled}
\newfloat{listing}{tb}{lst}{}
\floatname{listing}{Listing}
\title{Models as Agents: Optimizing Multi-Step Predictions of Interactive Local Models in Model-Based Multi-Agent Reinforcement Learning}
\author{
    Zifan Wu\textsuperscript{\rm 1}, 
    Chao Yu\textsuperscript{\rm 1,2}\thanks{Corresponding to yuchao3@mail.sysu.edu.cn.}, 
    Chen Chen\textsuperscript{\rm 3}, 
    Jianye Hao\textsuperscript{\rm 3}, 
    Hankz Hankui Zhuo\textsuperscript{\rm 1}
}
\affiliations{
    \textsuperscript{\rm 1}School of Computer Science and Engineering, Sun Yat-sen University, Guangzhou, Guangdong, China\\
    \textsuperscript{\rm 2}Pengcheng Laboratory, Shenzhen, Guangdong, China\\
    \textsuperscript{\rm 3}Huawei Noah's Ark Lab, Beijing, China
}

\usepackage{bibentry}
\begin{document}

\maketitle

\begin{abstract}
%model-free MARL的success -> model-based MARL工作较少 -> 为什么？引出model-based困难性 -> 引出local model的必要性 -> local model 存在的问题：多model交互，在multi-step prediction中问题被突出 -> 引出自己的idea
Research in model-based reinforcement learning has made significant progress in recent years.
Compared to single-agent settings, the exponential dimension growth of the joint state-action space in multi-agent systems dramatically increases the complexity of the environment dynamics, which makes it infeasible to learn an accurate global model and thus necessitates the use of agent-wise local models.
However, during multi-step model rollouts, the prediction of one local model can affect the predictions of other local models in the next step.
As a result, local prediction errors can be propagated to other localities and eventually give rise to considerably large global errors.
Furthermore, since the models are generally used to predict for multiple steps, simply minimizing one-step prediction errors regardless of their long-term effect on other models may further aggravate the propagation of local errors.
To this end, we propose Models as AGents (MAG), a multi-agent model optimization framework that reversely treats the local models as multi-step decision making agents and the current policies as the dynamics during the model rollout process.
In this way, the local models are able to consider the multi-step mutual affect between each other before making predictions.
Theoretically, we show that the objective of MAG is approximately equivalent to maximizing a lower bound of the true environment return.
Experiments on the challenging StarCraft II benchmark demonstrate the effectiveness of MAG.
\end{abstract}

\section{Introduction}
\label{sec:intro}
\noindent Model-Based Reinforcement Learning (MBRL)~\citep{moerland2020model, luo2022survey} aims to improve the sample efficiency of model-free methods by learning an approximate world model and then using it to aid policy learning.
Despite the success in single-agent settings, there are still limited works concentrating on MBRL in multi-agent systems.
In these systems, the exponential dimension growth of the joint state-action space dramatically increases the complexity of the environment dynamics, making it infeasible to learn an accurate global model~\citep{zhang2021model,wang2022model}.
Thus, a common practice is to make use of local agent-wise models which only require partial information and then predict the most relevant information for policy learning of each corresponding agent, so as to alleviate the issue of high dimension and avoid modelling the whole complicated dynamics~\citep{kim2020communication,zhang2021model}.
%这一段接下来讲multi-agent model困难性，引出local model

\begin{figure}[t]
\centering
\includegraphics[width=1.05\columnwidth,height=90pt]{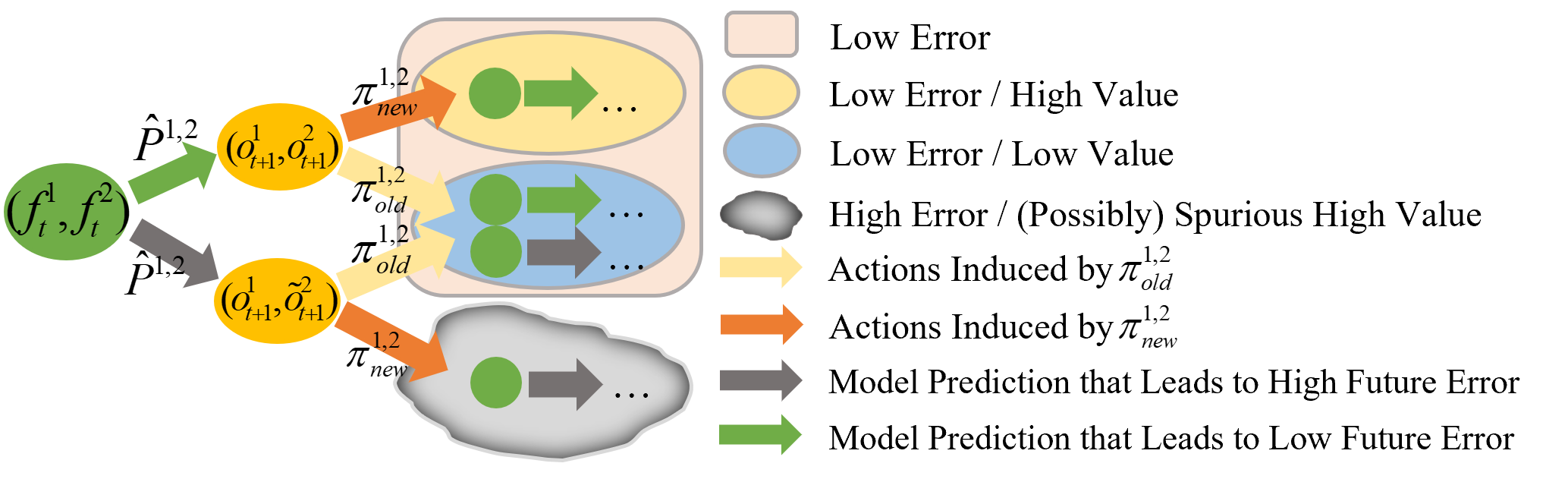} % Reduce the figure size so that it is slightly narrower than the column. Don't use precise values for figure width.This setup will avoid overfull boxes.
\caption{Intuition for multi-step interactions between the local models and the policies.}
\label{img:illu}
\end{figure}
One of the most commonly used paradigms in  Multi-Agent Reinforcement Learning (MARL) is Centralized Training with Decentralized Execution (CTDE)~\citep{lowe2017multi,foerster2018counterfactual,rashid2018qmix,wang2021qplex,wu2021coordinated,fu2022revisiting}, which allows the use of global information in the policy training phase, yet retains local observability of each agent during execution.
MAMBA~\citep{egorov2022scalable}, a recently proposed multi-agent MBRL method under the CTDE paradigm, achieves state-of-the-art sample efficiency in several challenging benchmarks, especially the StarCraft II challenge~\citep{samvelyan2019starcraft}.
To take full advantage of the centralized training phase, MAMBA utilizes the Attention mechanism~\citep{vaswani2017attention} to extract information for each local model from the global information, i.e., $(f^1,\ldots, f^N) = \text{Attention}(l^1, \ldots, l^N)$, where $N$ is the number of agents, $l^i:=(o^i, a^i)$ denotes the local information of agent $i$, $o^i$ and $a^i$ are the observation and action of agent $i$ respectively, and $f^i$ denotes the extracted feature for the local model of agent $i$.
However, since the Attention block fuses all local predictions obtained from the local models, the prediction of each local model, i.e., $\hat{P}^i(o^{i'}|f^i)$, can affect the subsequent predictions of other local models in the next rollout step.

Furthermore, while generally trained to simply minimize one-step prediction errors, the local models are usually not able to take into account the aforementioned multi-step errors induced by the interactions between the local models and policies. 
As a result, prediction errors of one local model can be propagated to the others and eventually induce large accumulative global errors during multi-step model rollouts, which would hinder the learning of policies.
% In addition, from the perspective of empirical effects, MAMBA's improvement in sample efficiency is limited and far less than single-agent like MBPO in MuJoCo, and is hard to win in some maps with high difficulty. 

%本段用图例直观说明motivation的合理性
Figure~\ref{img:illu} gives an intuitive illustration of the above discussion: 
Given local features $(f_t^1, f_t^2)$ w.r.t. agent $1$ and $2$ at step $t$,
1) the local models predict  $\hat{P}^1(o^1_{t+1}|f_t^1)=1$ and $\hat{P}^2(o^2_{t+1}|f_t^2)=\hat{P}^2(\tilde{o}^2_{t+1}|f_t^2)=50\% $;
and 2) under the previous joint policy $\pi_{\text{old}}^{1, 2}$, both prediction of the next joint observation, i.e., $(o^1_{t+1},o^2_{t+1})$ and $ (o^1_{t+1}, \tilde{o}^2_{t+1})$, will lead the trajectory to go into regions with low values predicted by the value function, hence $\pi_{\text{old}}^{1, 2}$ is updated to $\pi_{\text{new}}^{1, 2}$ to explore regions with potential high values. 
But under the updated joint policy $\pi_{\text{new}}^{1,2}$, the subsequent rollout trajectory starting from $(o^1_{t+1}, \tilde{o}^2_{t+1})$ would lead to considerably larger model errors compared to the trajectory starting from $(o^1_{t+1}, o^2_{t+1})$.
Thus, to reduce accumulative model errors along rollout trajectories, the local models should learn to coordinate with each other while quickly adapting to the current joint policy.
Formally, as will be shown in Section~\ref{sec:method}, smaller accumulative model errors could provide stronger performance guarantee.

%本段简述方法和优点（为什么要这么做），简述理论和实验
In this work, we propose Models as AGents (MAG), a multi-agent model optimization framework which considers the interactions between local models during multi-step model rollout.
Based on the MAMBA framework, the whole environment dynamics is decomposed into agent-wise local models, and our key idea lies in reversely considering the local models as multi-step decision makers while fixing the current joint policy to serve as the environment.
During model learning, the local models perform multi-step interactions with each other as well as the policies, so as to take the long-term global effect of immediate local predictions into account and generate trajectories with less accumulative errors.
Theoretically, we show the necessity of considering the local model interactions and minimizing the multi-step accumulative errors.
Empirically, the results on several challenging tasks in the StarCraft II benchmark demonstrate that MAG significantly outperforms MAMBA in low data regime, and the model error analysis further verifies the effectiveness of our model learning mechanism.

\section{Background}
\label{sec:back}
In this section, we first introduce the problem setting of MARL and MBRL, and then give a brief description of MAMBA, the aforementioned state-of-the-art model-based MARL method.
% \subsection{}
\paragraph{MARL}
In this work, we focus on the fully cooperative multi-agent systems that can be formalized as Dec-POMDPs~\citep{oliehoek2016concise}, which are defined by tuple $(N, S,  \Omega, O, A, R, P, \gamma)$, 
where $N$ is the number of agents, 
$S$  the set of global states, 
$\Omega$ the observation space shared by the agents, 
$O(s, i)$ the function deriving partial observations for each agent $i$ from a global state $s\in S$, 
$A$ the action space, 
$R(s, a^1, \ldots, a^N)$ a shared scalar reward function that takes $s\in S$ and $a^i\in A, i\in\{1,\ldots,N\}$ as input,
and $\gamma\in[0, 1)$ the discount factor.
Each agent has an action-observation history $\tau^i\in T\equiv(\Omega\times A)^*$.
We use the bold symbol $\boldsymbol{o},\boldsymbol{a}, \boldsymbol{\pi}$ to denote the joint observation $\{o^1,\ldots,o^N\}$, action $\{a^1,\ldots,a^N\}$ and policy $\{\pi^1,\ldots,\pi^N\}$, respectively.
At each timestep, agent $i$ chooses an action $a^i\in A$ according to its policy $\pi^i(a^i|\tau^i)$ (We replace $\tau^i$ by $o^i$ in our analysis for brevity).
The environment then returns the reward signal $R(s, \boldsymbol{a})$ and shifts to the next state according to the transition function $P(s'|s, \boldsymbol{a})$.
The expected return of joint policy $\boldsymbol{\pi}$ is defined by $J(\boldsymbol{\pi}) := \mathbb{E}_{\boldsymbol{\pi}}[\sum_{t'=0}^{\infty}\gamma^{t'}R_{t+t'}|s_t, \boldsymbol{a}_t]$.
Some previous works~\citep{wang2020few, liu2020multi} have shown that it is possible to significantly reduce the state space in large environments to only relevant information for the agents’ decision making. Hence, in this paper we assume that the joint observation-action, i.e.,  $(\boldsymbol{o}, \boldsymbol{a})$, is sufficient to predict the next joint observation $\boldsymbol{o}'$ and the global reward $R$.
Serving as a special case of Dec-POMDPs, MMDPs~\citep{boutilier1996planning} assume global observability of each agent and are adopted to reformulate the model rollout process in our work.

% \subsection{}
\paragraph{MBRL}
MBRL methods learn a model $\hat{P}$ that approximates the unknown dynamics $P$, and then use this model to assist policy learning.
While the model can be utilized in various ways~\citep{ feinberg2018model,d2020learn,curi2020efficient,song2021pc, amos2021model}, this work focuses on one of the most common usages, i.e., generating pseudo samples to enrich the dataset, so as to accelerate policy learning and reduce interactions with the true environment~\citep{sutton1991dyna,chua2018deep,luo2018algorithmic,janner2019trust}.
The expected return of policy $\boldsymbol{\pi}$ predicted by model $\hat{P}$ is denoted as $J^{\hat{P}}(\boldsymbol{\pi}):=\mathbb{E}_{\boldsymbol{\pi}, \hat{P}}[\sum_{t'=0}^{\infty}\gamma^{t'}R_{t+t'}|s_t, \boldsymbol{a}_t]$.
As a state-of-the-art MBRL method in discrete environments, Dreamer V2~\citep{hafner2020mastering} makes use of the RSSM model~\citep{hafner2019learning} to learn the dynamics of the environment in the latent space by minimizing the evidence lower bound~\citep{kingma2013auto}.

\label{sec:mamba}
\paragraph{MAMBA}
Building upon Dreamer V2, MAMBA~\citep{egorov2022scalable} also learns the environment dynamics in the latent space, and makes use of the Attention mechanism~\citep{vaswani2017attention} to extract features for each local models from global information.
To disentangle the agents' latent space and encourage the local models to be mutually independent when making predictions, MAMBA proposes to maximize the mutual information between the latent state and the previous action of the corresponding agent.
In addition, the method allows communicating with the neighbouring agents via discrete messages to sustain world models during the execution phase, thus regarding world models as an instance of communication.
To the best of our knowledge, MAMBA is the first model-based MARL method that improves the sample efficiency of model-free methods by an order of magnitude on the challenging StarCraft II benchmark.
Nevertheless, compared to the performance of Dreamer V2 in Atari games~\citep{bellemare2013arcade} and MBPO~\citep{janner2019trust} in the MuJoCo~\citep{todorov2012mujoco} benchmark, the overall improvement of sample efficiency, as well as the asymptotic performances in some difficult tasks achieved by MAMBA are still relatively limited, which may be due to the high complexity of the dynamics of multi-agent systems.

% Without loss of generality, we denote the local model by $\hat{P}^i(o^{i'}|\boldsymbol{o}, \boldsymbol{a})$ in our theoretical analysis for brevity.

\section{Method}
\label{sec:method}
In this section, we first propose a theoretical result of how the prediction errors of agent-wise local models affect the overall policy performance, based on which we reformulate the model rollout process as a multi-agent sequential decision making problem.
In the last subsection, we present the practical implementation of MAG and further detail some important steps in the algorithm.

\subsection{Theoretical Result}
\label{sec:theory}
%单智能体已有工作证明过bound，但由于对模型误差取了max_t，所以在多智能体模型学习中这一点会导致忽略模型之间的交互（交互需要多步来体现）导致累积误差非常大，进而
Since general MBRL methods optimize the policy by maximizing the expected return predicted by the model, one of the most crucial theoretical problems for MBRL is to bound the gap between the model predicted return and the true environment return.
% Formally, suppose that this gap can be upper bounded by $C$, and the policy is updated from $\boldsymbol{\pi}_D$ to $\boldsymbol{\pi}$ with the expected model return improved by $2C$, 
% then the policy improvement in terms of the true environment can be guaranteed:
% \begin{align}
% J(\boldsymbol{\pi})-J(\boldsymbol{\pi}_D)&\geq \left(J^{\hat{P}}(\boldsymbol{\pi})-C\right) - \left(J^{\hat{P}}\left(\boldsymbol{\pi}_D\right)+C\right)\nonumber\\
% &\geq J^{\hat{P}}(\boldsymbol{\pi})-J^{\hat{P}}(\boldsymbol{\pi}_D)- 2C\geq 0. \nonumber
% \end{align}
Our major theoretical result is the following theorem that bounds the performance gap:
\begin{theorem}
\label{thm:upperbound}
Denoting the set of local models by $\hat{P}:=\{\hat{P}^i \}_{i=1}^N$ and the data-collecting policy obtained in the last iteration by $\boldsymbol{\pi}_D$, the gap between the expected return of the model and the environment can be bounded as~\footnote{In our theoretical analysis, the reward function is assumed to be known. Note that this is a commonly adopted assumption since the sample complexity of learning the reward function with supervised learning is a lower order term compared to the one of learning the transition model~\citep{gheshlaghi2013minimax}.}:
    \begin{align}
    \label{eq:upperbound}
\left| J(\boldsymbol{\pi})-J^{\hat{P}}(\boldsymbol{\pi})\right| \leq
\frac{R_{max}}{(1-\gamma)^2}\left(2\epsilon_{\boldsymbol{\pi}}+
(1-\gamma)\sum_{t=1}^\infty\gamma^t\epsilon_{m_t}\right),
\end{align}
where $\epsilon_{\boldsymbol{\pi}}:=\max_{\boldsymbol{o}} D_{TV}(\boldsymbol{\pi}_D(\cdot| \boldsymbol{o})\|\boldsymbol{\pi}(\cdot| \boldsymbol{o}))$ denotes the distribution shift of the joint policy between two consecutive iterations,
$\epsilon_{m_t}:=\mathbb{E}_{\boldsymbol{o}\sim \hat{P}_{t-1}(\cdot;\boldsymbol{\pi})} \left[\max_{\boldsymbol{a}}\sqrt{2\sum_{i=1}^N\mathbb{E}_{\boldsymbol{o}'\sim \hat{P}(\cdot|\boldsymbol{o}, \boldsymbol{a})}\left[ \log\frac{\hat{P}^i(o^{i'}|\boldsymbol{o}, \boldsymbol{a})}{\sqrt[N]{P(\boldsymbol{o}'|\boldsymbol{o}, \boldsymbol{a})}}   \right]}\right] $ denotes the upper bound of the $i$-th model's error at timestep $t$ of the model rollout trajectory,
$\hat{P}_{t-1}(\boldsymbol{o};\boldsymbol{\pi})$ denotes the distribution of joint observation at $t-1$ under $\hat{P}$ and $\boldsymbol{\pi}$, 
and $R_{max}:=\max_{s, \boldsymbol{a}} R(s, \boldsymbol{a})$.
\end{theorem}
\begin{proof}
Please refer to Appendix~A.
\end{proof}
%解读bound：先解读多步、当前策略与model交互、多local model交互，与此前单智能体的bound的不同
%讲由此应该怎么做算法，后面顺接第二节framework
% In our theoretical analysis, we assume $\hat{P}(\boldsymbol{o}'|\boldsymbol{o}, \boldsymbol{a})= \prod_{i=1}^N \hat{P}^i(o^{i'}|\boldsymbol{o}, \boldsymbol{a})$ since the local models can be approximately independent when making predictions by maximizing the mutual information between the latent state and the previous action of the corresponding agent, which are detailed in the MAMBA implementation.
It is worth noting that Theorem~\ref{thm:upperbound} is not simply a multi-agent version of the results that have been derived in the single-agent setting~\citep{luo2018algorithmic,janner2019trust}.
The key difference is that Theorem~\ref{thm:upperbound} does not scale up the step-wise model prediction errors (i.e., $\epsilon_{m_t}$) to their maximum over timesteps, which not only leads to a tighter bound that provides stronger guarantee for policy improvement (see Appendix~A for proof), but also indicates how the interactions between local models affect the overall performance error bound:
Note that by definition the model error at step $t$, i.e., $\epsilon_{m_t}$, depends on the distribution of the joint observation at the last timestep, i.e., $\hat{P}_{t-1}(\boldsymbol{o};\boldsymbol{\pi})$, and except for the first step of rollout trajectories, this distribution further depends on the current policies $\boldsymbol{\pi}$ and the prediction of other local models at the last timestep, i.e., $\hat{P}_{t}(\boldsymbol{o}; \boldsymbol{\pi})=\mathbb{E}_{\boldsymbol{o}\sim \hat{P}_{t-1}(\cdot;\boldsymbol{\pi}), \boldsymbol{a}\sim\boldsymbol{\pi}(\cdot|\boldsymbol{o}))}[\prod_{i=1}^N\hat{P}^i(o^i|\boldsymbol{o}, \boldsymbol{a})]$.
Thus, the errors of the local models can affect each other during multi-step model rollout, and this mutual affect can largely determine the tightness of the overall error bound.

Based on this result, a performance lower bound with regard to the policy shift and the model error can be written as: 
$J(\boldsymbol{\pi})\geq J^{\hat{P}}(\boldsymbol{\pi}) - 2C(\hat{P}, \boldsymbol{\pi})$, where $C(\hat{P}, \boldsymbol{\pi})$ denotes the right hand side of Eq.~(\ref{eq:upperbound}).
Then, in an ideal manner, applying the following update rule repeatedly can \textbf{guarantee the monotonic improvement} of the joint policy:
\begin{align}
\label{eq:monotonic}
    \hat{P}, \boldsymbol{\pi} \leftarrow \mathop{\arg\max}_{\hat{P}, \boldsymbol{\pi}} J^{\hat{P}}(\boldsymbol{\pi}) - C(\hat{P}, \boldsymbol{\pi}).
\end{align}
The update rule in Eq.~(\ref{eq:monotonic}) is often impractical since it involves an exhaustive search in the joint state-action space to compute $C$, and requires full-horizon rollouts in the model for estimating the accumulative model errors.
Thus, similar to how algorithms like TRPO~\citep{schulman2015trust} approximate their theoretically monotonic version, this update rule can be approximated by maximizing the expected model return (i.e., $J^{\hat{P}}(\boldsymbol{\pi})$) while keeping the accumulative model error (i.e., $\sum_{ t}\gamma^t\epsilon_{m_t}$) small.
As for the policy shift term $\epsilon_\pi$, though the bound suggests that this term should also be constrained, we found empirically that it is sufficient to only control the model error.
This may be explained by the relatively small scale of policy shift w.r.t. the model error, as observed in~\citep{janner2019trust}.

By treating $-\epsilon_{m_t}$ as the ``reward'' shared by the local models at timestep $t$, the learning of the local models can be regarded as an optimization process of multi-step predictions of the local models, where the objective is to minimize the global prediction errors accumulated along the model rollout trajectories.
Note that the definition of $\epsilon_{m_t}$ involves the expectation under the current joint policy $\boldsymbol{\pi}$, thus during model learning, the joint policy can be fixed to serve as a background environment, while the local models reversely play the role of decision-makers that should learn to maximize the ``expected return'' (i.e., $-\mathbb{E}_{\boldsymbol{\pi}, \hat{P}}[\sum_{ t}\gamma^t\epsilon_{m_t}]$) under the current joint policy.
Building on the above theoretical intuition, we now propose the MAG framework in the next subsection.
%引出下一节的framework

\subsection{Problem Reformulation}
\label{sec:reformulation}
To formalize the intuition of reversing the roles of the models and the agents during model learning, we first define the \textit{model MMDP} to reformulate the model rollout process and then outline the overall model optimization of MAG as a generic solution to the reformulated problem.
\begin{defi}
The \textit{model MMDP} is defined by tuple $(N, \gamma, S_m, A_m, P_m, R_m)$, where $N$ is the number of local models, $\gamma$ is the discount factor, $S_m, A_m, P_m$ and $R_m$ are the model-state space, the model-action space, the model-transition function and the scalar model-reward function, respectively. 

At each timestep $t$, each local model $\hat{P}^i$ receives model-state $s_{m_t}:=(\boldsymbol{o}_t, \boldsymbol{a}_t)\in S_{m_t}$, then takes a model-action $a^i_{m_t}:=o^i_{t+1}$ according to its ``policy'' $\hat{P}^i(a^i_{m_t}|s_{m_t})$. 
After that, the model-transition function returns the next model-state by $P_m(s_{m_{t+1}}|s_{m_t}, \boldsymbol{a}_{m_t}):=\boldsymbol{\pi}(\boldsymbol{a}_{t+1}|\boldsymbol{o}_{t+1})\prod_{i=1}^N\hat{P}(o^i_{t+1}|\boldsymbol{o}_t, \boldsymbol{a}_t)$, while the model-reward function returns a scalar reward by $R_m(s_{m_t}, a_{m_t}):=\sum_{i=1}^N\log\frac{\hat{P}^i(o^i_{t+1}|\boldsymbol{o}_t, \boldsymbol{a}_t)}{\sqrt[N]{P(\boldsymbol{o}_{t+1}|\boldsymbol{o}_t, \boldsymbol{a}_t)}}$.
\end{defi}
Using the \textit{model MMDP} formulation, the model learning phase can be viewed as a multi-agent learning problem, where the current joint policy is fixed to serve as the environment dynamics and the local models, now as the decision makers, interact with each other and learn to minimize the accumulative prediction error under the current joint policy.
From this perspective, the local models trained by minimizing one-step prediction errors for each individual can be intuitively interpreted as greedy independent learners, which are often considered shortsighted and may struggle to learn cooperative behaviors.
To minimize the accumulative global errors, the local models must instead consider the long-term global effect of immediate local predictions.

Note that in the competitive or mixed cooperative-competitive scenarios, the goal of each local model is generally to assist policy learning of only one individual agent, thus in those scenarios the local models would aim at minimizing the individual accumulative errors instead of the global summation of model errors.
Consequently, in those scenarios, the model rollout process can be defined as a \textit{Markov Game}~\citep{shapley1953stochastic}, where the reward function can be defined respectively for each local model.
Since the major focus of this work is the fully cooperative scenarios, we leave the above discussion as a possible motivation for future work.

Similar to the optimization of the policy, the objective of model learning can be written as $\arg\max_{\hat{P}} J^{\boldsymbol{\pi}}\big(\hat{P}\big)$ where $J^{\pi}\big(\hat{P}\big) :=\mathbb{E}_{\boldsymbol{\pi}, \hat{P}}[\sum_t\gamma^t R_m(s_{m_t}, \boldsymbol{a}_{m_t})] $.
Due to this duality between the learning of the policies and the models, we call this overall model-based MARL method by Models as AGents (MAG).
Specifically, during model learning, the local models first generate samples by actively interacting with the current joint policy (now viewed as the background environment), and then optimize the expected return $J^{\boldsymbol{\pi}}\big(\hat{P}\big)$ accordingly.
% The practical implementation to optimize $J^{\boldsymbol{\pi}}\big(\hat{P}\big)$ is detailed in the next subsection.

\subsection{Practical Implementation}
\label{sec:practical}
% According to the problem reformulation in the last subsection, the local models can be regarded as multi-step decision-makers that must learn to cooperate with each other.
To give a practical solution to the \textit{model MMDP}, we describe the implementation of MAG in this subsection.
% It is worth noting that our implementation is built on the MAMBA method which uses the Attention mechanism to extract partial information for each local model, and adopts Dreamer V2 to learn the dynamics in the latent space.
% As a result, the entire implementation is complicated so some details are simplified here for readability.
\begin{algorithm}[t]
\caption{MAG}
\label{alg:mpc}
\begin{algorithmic}[1] %[1] enables line numbers
\STATE Initialize joint policy $\boldsymbol{\pi}$, predictive local models $\hat{P}^i, i=1, 2,\ldots,N$, model-reward predictor $\hat{R}_m$, environment dataset $\mathcal{D}_e$ and model dataset $\mathcal{D}_m$.
\FOR{$N$ episodes}
\STATE Collect an episode of real-environment data using $\boldsymbol{\pi}$ and then add the data to $\mathcal{D}_e$;
\STATE Train models $\{\hat{P}^i\}_{i=1}^N$ on $\mathcal{D}_e$ via one-step prediction loss;
\STATE Train $\hat{R}_m$ on $\mathcal{D}_e$ via supervised learning on $\mathcal{D}_e$;

\FOR{$M$ model rollouts}
\STATE Sample joint observations $\boldsymbol{o}$ uniformly from $\mathcal{D}_e$ and use them as the initial observations for our rollout trajectories;
\FOR{$k$ rollout steps}
\STATE Each agent takes action $a^i$ according to $\pi^i(\cdot|o^i)$;
\STATE Initialize $s_{m_0}=(\boldsymbol{o}, \boldsymbol{a})$ and perform $L$ parallelized rollouts for $H$ steps by actively interacting the local models with $\boldsymbol{\pi}$: $\{s_{m_{0, j}} , \boldsymbol{a}_{m_{0, j}}, R_{m_{0, j}}, s_{m_{1, j}}, \boldsymbol{a}_{m_{1, j}},\ldots, s_{m_{H, j}} \}_{j=1}^L$, where $s_{m_{0, j}}=s_{m_0}, \forall j\in\{1, 2, \ldots, N\}$;
\STATE Compute $r_{m_j}=\sum_{t=0}^{H-1}\hat{R}_m(s_{m_{t, j}}, \boldsymbol{a}_{m_{t, j}})$ for each rollout trajectory $j$;
\STATE Take $(\boldsymbol{o}', \boldsymbol{R}) = \boldsymbol{a}_{m_{0, \mathop{\arg\max}_j r_{m_j}}}$;
\STATE Store $(\boldsymbol{o}, \boldsymbol{a}, \boldsymbol{R}, \boldsymbol{o}')$ to $\mathcal{D}_m$ and then let $\boldsymbol{o}=\boldsymbol{o}'$;
\ENDFOR
\ENDFOR
\FOR{$G$ gradient updates}
\STATE Update $\pi$ using data sampled from $\mathcal{D}_m$;
\ENDFOR
\ENDFOR
\end{algorithmic}
\end{algorithm}

\paragraph{The Overall Algorithm}
Algorithm~\ref{alg:mpc} gives the overall algorithm design of MAG.
In each outer loop, the current joint policy is applied in the real environment to collect an episode of real-world data, which is then added to the environment dataset $\mathcal{D}_e$ (Line 3).
Then, the local models are pre-trained by traditional one-step prediction loss $\sum_{i=1}^N\|\hat{o}^{i'}-o^{i'}\|+\|\hat{R}^i-R\|$, where $\hat{o}^{i'},\hat{R}^i\sim \hat{P}^i(\cdot,\cdot|\boldsymbol{o}, \boldsymbol{a})$ and each transition $(\boldsymbol{o}, \boldsymbol{a}, R, \boldsymbol{o}')$ is sampled from the environment dataset (Line 4).
Since the reward function $R$ is generally not available in practice, each local model is also trained to predict the global reward respectively given $(\boldsymbol{o}, \boldsymbol{a})$.
Besides, it deserves to be noted that we do not directly use the above pre-trained local models to obtain the predictions during model rollout, but instead optimize the multi-step predictions of local models via a planning process.
In Line 5, MAG trains the $\hat{R}_m$ network to approximate $R_m$, since by definition the model-reward $R_m$ involves the true environment dynamics and thus cannot be directly computed.
The approximation of $R_m$ will be detailed later on.
Lines 6-15 give the model rollout process where $M$ parallelized trajectories of length $k$ are generated based on different initial observations sampled from $\mathcal{D}_e$.
For each rollout step, before predicting the next observation, the local models first treat the current joint policy as the ``dynamics'' and then perform $H$-step ($H\leq k$) planning to obtain the best predictions for the current step (Lines 10-12).
This is the core of MAG and will be detailed later.
Finally, the pseudo samples generated by the model are added to the model dataset, which is then used for policy learning.
Specifically, we adopt PPO~\citep{schulman2017proximal} as the underlying policy optimization method and use global information that has been processed by the Attention block as the input of the critic.

\begin{figure*}[t]
\centering
\includegraphics[width=0.98\textwidth]{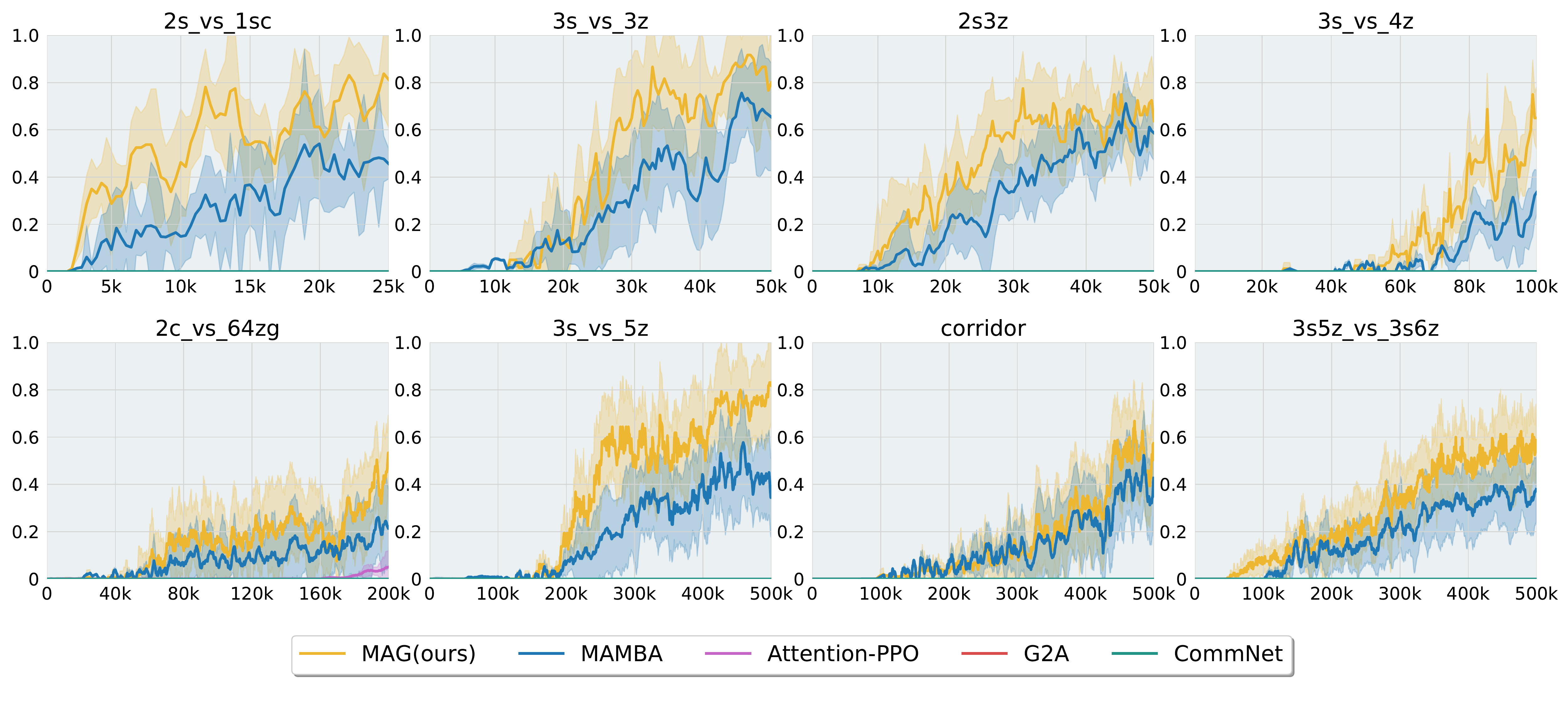} 
\caption{Comparisons against baselines on SMAC.
  Solid curves represent the mean of runs over 5 different random seeds, and shaded regions correspond to standard deviation among these runs. X axis denotes the number of steps taken in the real environment and Y axis denotes the win rate.}
\label{fig:overall_result}
\end{figure*}
\paragraph{Approximating $R_m$}
We approximate $R_m$ by training a neural network $\hat{R}_m$ which takes transitions $(\boldsymbol{o}, \boldsymbol{a}, R, \boldsymbol{o}')$ sampled from the environment dataset as inputs, and the model prediction errors on these transitions as labels. 
The prediction error of an environment transition is computed via $\sum_{i=1}^N\|\hat{o}^{i'}-o^{i'}\|+\|\hat{R}^i-R\|$, where $\hat{o}^{i'}$ and $\hat{R}^i$ are sampled from $\hat{P}^i(\cdot, \cdot|\boldsymbol{o}, \boldsymbol{a})$. 
Intuitively, $\hat{R}_m$ can be seen as an indicator that informs the models where their ``weaknesses'' lie in.
Additionally, since Dreamer V2 utilizes VAE~\citep{kingma2013auto} and learns the dynamics in a latent space, the actual loss of the dynamics consists of a reconstruction loss of the auto-encoder and a KL divergence loss that aim to minimize the distance between the prior and the posterior of the latent state.
Consequently, computing the model errors in Dreamer V2 can be interpreted as computing the prediction errors in the latent space, and thus is equivalent to computing $\sum_{i=1}^N\|\hat{o}^{i'}-o^{i'}\|+\|\hat{R}^i-R\|$ in principle.

\paragraph{Planning to Predict}
Since in our problem formulation the ``dynamics'' of the model rollout process (i.e., the current joint policy $\boldsymbol{\pi}$) is accessible, one of the simplest yet effective approaches to learn the models can be the Model Predictive Control (MPC)~\citep{camacho2013model}, which utilizes the dynamics to plan and optimize for a sequence of actions. 
Given the state $s_{m_t}$ at step $t$, the MPC controller first optimizes the sequence of actions $\boldsymbol{a}_{m_{t:t+H}}$ over a finite horizon $H$, and then employs the first action of the optimal action sequence, i.e., $\boldsymbol{a}_{m_{H, t}}:=\mathop{\arg\max}_{\boldsymbol{a}_{m_{t:t+H}}}\mathbb{E}_{\hat{P}, \boldsymbol{\pi}}\sum_{t'=t}^{t+H-1}R_m(s_{m_{t'}}, \boldsymbol{a}_{m_{t'}}) $.
Computing the exact $\mathop{\arg\max}_{\boldsymbol{a}_{m_{t:t+H}}}$ requires a complete search in a space of dimension $|A_m|^{N\cdot H}$, which is  impractical in most scenarios.
Thus, as specified from Lines 10-12 in Algorithm~\ref{alg:mpc}, we adopt the random-sampling shooting method~\citep{rao2009survey} which generates $L$ random action sequences, executes them respectively, and chooses the one with the highest return predicted by the dynamics.
Essentially, this planning process is a simulation of the interactions between the local models and the current joint policy, according to which each local model chooses the best prediction that approximately minimizes the global model error in concert with the other local predictions, thus achieving the coordination between local models.

\section{Experiments}
In this section, we present an empirical study of MAG on the challenging StarCraft II benchmark (SMAC)~\citep{samvelyan2019starcraft}.
In the first subsection, we provide the overall comparison between MAG and several baselines.
Then, we provide a quantitative analysis on the multi-step prediction loss to verify the effectiveness of our algorithm design in model learning.
In the last subsection, we conduct ablation studies to show how the choices of the planning horizon (i.e., $H$ in Algorithm~\ref{alg:mpc}) and the number of random shooting trajectories (i.e., $L$ in Algorithm~\ref{alg:mpc}) affect the overall performance.

\subsection{Comparative Evaluation}
\paragraph{Baselines} We compare MAG with a model-based baseline and several model-free baselines.
The model-based baseline is MAMBA~\citep{egorov2022scalable}, a recently proposed multi-agent MBRL method that achieves state-of-the-art sample efficiency in several SMAC tasks.
The model-free baselines include 1) Attention-PPO, the model-free counterpart of both MAG and MAMBA which equips PPO~\citep{schulman2017proximal} with centralized attention-critics and communication during execution; 2) G2A~\citep{liu2020multi}, which adopts a two-stage attention architecture to realize communication between agents; and 3) CommNet~\citep{sukhbaatar2016learning}, which applies LSTM~\citep{hochreiter1997long} to learn continuous communication protocols for partially observable environments.
In addition, it deserves to note that MAG is essentially a flexible plug-in component which can be employed by most model-based methods to improve the learning of the model.
In our comparisons, we plug the model learning process of MAG into MAMBA.

\paragraph{Environments} The methods are evaluated on 8 maps of SMAC, ranging from \textit{Easy} maps (2s\_vs\_1sc, 2s3z, 3s\_vs\_3z), \textit{Hard} maps (3s\_vs\_4z, 3s\_vs\_5z, 2c\_vs\_64zg) and \textit{Super Hard} maps (corridor, 3s5z\_vs\_3s6z).

\paragraph{Implementation Details} The implementation of MAG is overall built on MAMBA.~\footnote{Code available at https://github.com/ZifanWu/MAG.} %and the only incremental network component is the $\hat{R}_m$ network.
For more details of the hyperparameter settings, please refer to Appendix~B.

\paragraph{Results} The overall results shown in Figure~\ref{fig:overall_result} demonstrate that MAG consistently outperforms all the baselines in low data regime.
The comparison between MAG and MAMBA verifies the effectiveness of optimizing multi-step prediction errors that are induced by the interactions between local models.
Besides, note that except for Attention-PPO in 2c\_vs\_64zg, all model-free baselines fail to even achieve a non-zero win rate in such low data regimes, showing the significant improvement of sample efficiency resulted from using a world model.

\subsection{Model Error Analysis}
Based on the theoretical result presented in Section~\ref{sec:theory}, the core idea of MAG is to reverse the roles played by the local models and the current joint policy, thus treating the models as decision-makers interacting with each other and aiming at minimizing the global accumulative model error. 
To validate the effectiveness of this algorithmic design, we empirically study the accumulative prediction error on the 2c\_vs\_64zg map.
While the real dynamics is unavailable during training, the error is approximated by a neural network trained on the environment dataset, i.e., $\hat{R}_m$.

\begin{figure}[h]
\centering
\includegraphics[width=0.94\columnwidth]{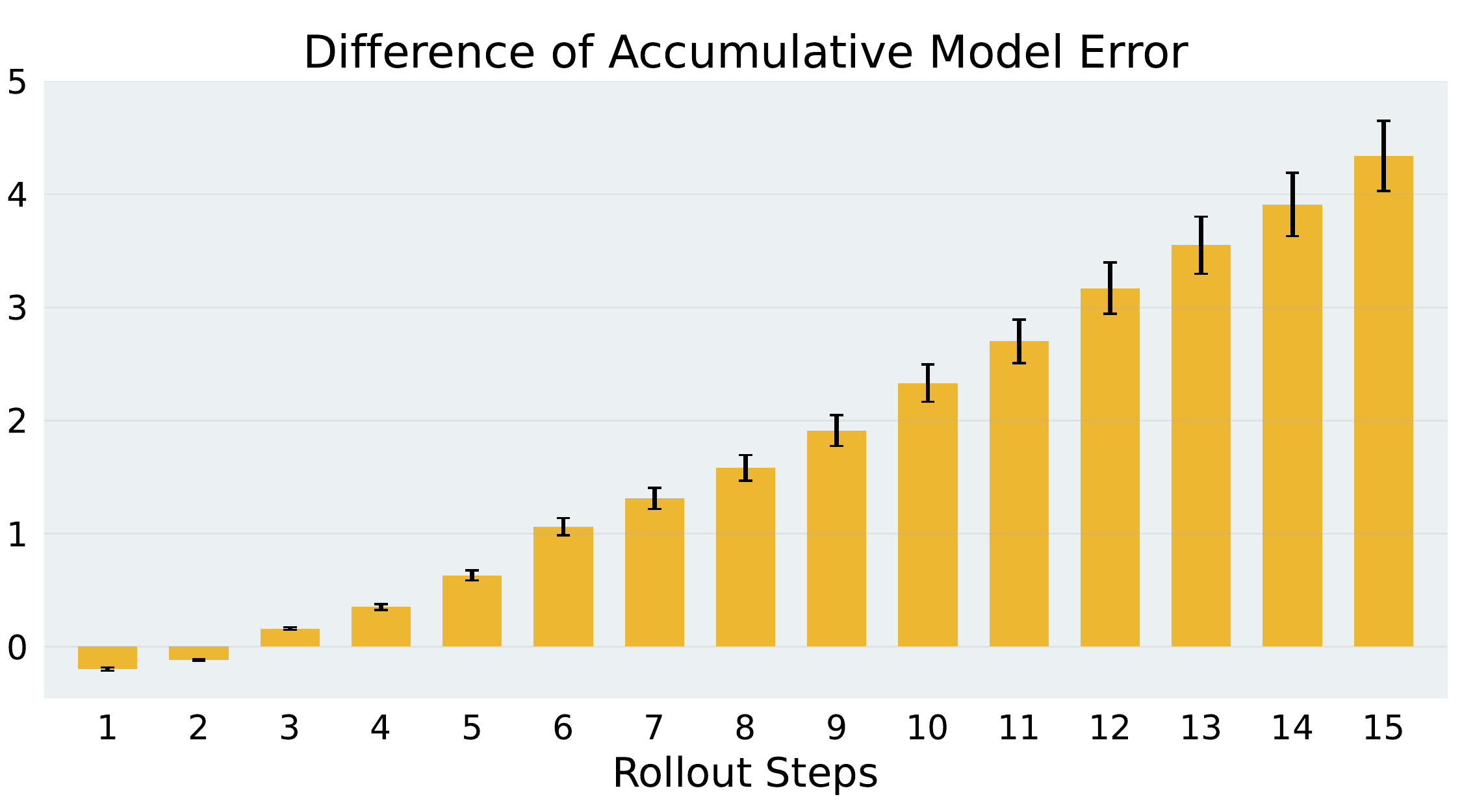} % Reduce the figure size so that it is slightly narrower than the column. Don't use precise values for figure width.This setup will avoid overfull boxes.
\caption{The difference of the accumulative model errors between MAMBA and MAG (the accumulative errors of MAMBA minus the accumulative errors of MAG), on the 2c\_vs\_64zg map.}
\label{fig:error_analysis}
\end{figure}

The result in Figure~\ref{fig:error_analysis} demonstrates that as the model rollout trajectories go longer, the accumulative model errors of MAMBA become significantly larger than that of MAG, which not only validates the effectiveness of MAG in reducing the accumulative model errors, but also provides a solid support for our theoretical result derived in Section~\ref{sec:method}, i.e., the method inducing less accumulative error is likely to achieve better performance.
Besides, we can also observe that in the first two steps the model errors induced by MAG are slightly larger than the errors of the baseline. 
This further agrees with the intuition mentioned in Section~\ref{sec:method} by showing that MAG is able to trade the one-step greedy model error for the accumulative error by considering the long-term effect of the immediate prediction.

\subsection{Ablation Studies}

According to the descriptions in Algorithm~\ref{alg:mpc}, apart from easy implementation, another advantage of utilizing MPC in optimizing the accumulative model-reward is that this only introduces a small number of extra hyperparameters.
Specifically, there are mainly two extra hyperparameters that need to be tuned, i.e., the planning horizon $H$ and the number of random shooting trajectories $L$.
The ablation results of these two hyperparameters are shown in Figure~\ref{fig:abla_H} and Figure~\ref{fig:abla_L} respectively, which indicate that longer planning horizons and more random trajectories always induce better performance.
Since increasing $H$ and $L$ leads to a rapid growth in terms of the computational complexity and the memory cost, the finally adopted settings of the two hyperparameters, which are detailed in Appendix~B, can be regarded as a compromise between this practical limit and the performance.

\begin{figure*}[t]
\centering
\includegraphics[width=1\textwidth]{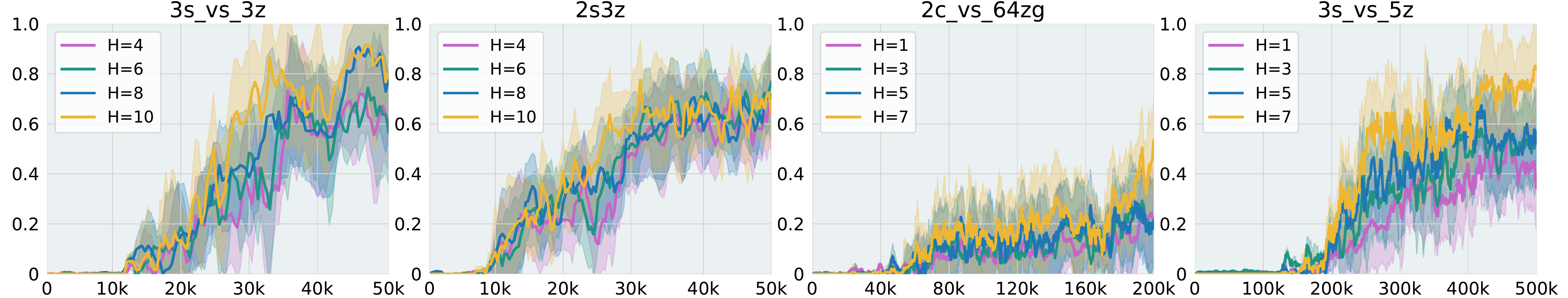} % Reduce the figure size so that it is slightly narrower than the column.
\caption{Ablation study of the planning horizon $H$.}
\label{fig:abla_H}
\end{figure*}
\begin{figure*}[t]
\centering
\includegraphics[width=1\textwidth]{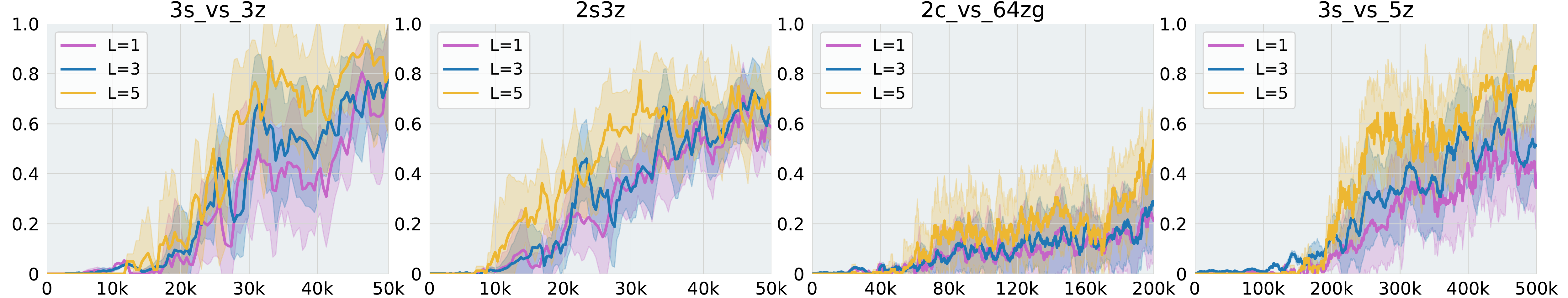} % Reduce the figure size so that it is slightly narrower than the column.
\caption{Ablation study of the the number of random shooting trajectories $L$.}
\label{fig:abla_L}
\end{figure*}

\section{Related Works}
The research of MBRL can be roughly divided into two lines: the model usage and the model learning. 
This work focuses on model learning and adopts the most common model usage, that is, generating pseudo samples to enrich the data buffer, so as to reduce the interaction with the environment and  accelerate policy learning~\citep{sutton1990integrated,sutton1991dyna,deisenroth2013survey,kalweit2017uncertainty,luo2018algorithmic,janner2019trust,pan2020trust}. 
Most of previous works in MBRL train the model simply by minimizing each one-step prediction error for transitions available in the environment dataset~\citep{kurutach2018model,chua2018deep,janner2019trust}.
However, in the multi-agent setting, the dimension of the joint observation-action space grows rapidly w.r.t. the number of agents, making it impractical to learn a global model for such complex environments~\citep{zhang2021model,wang2022model}.
% Some previous works have shown that it is possible to significantly reduce the state space in large environments to only include relevant information for the agents’ decision making~\citep{wang2020few, liu2020multi}. 
Thus, a common approach is to train a local model for each agent which takes partial observations as input and predicts relevant information for the agent's policy~\citep{kim2020communication,zhang2021model}.
% Moreover, in cooperative MARL, global information is available during the centralized training phase in the CTDE paradigm~\citep{gronauer2022multi}.
Therefore, MAMBA~\citep{egorov2022scalable} proposes to extract relevant information for each local model from the global information via the Attention mechanism~\citep{vaswani2017attention}, so as to avoid modelling the whole complicated dynamics and accelerate the model learning.
% The detailed description of MAMBA can be found in the Section~\ref{sec:back}.
Although the Attention mechanism is effective in extracting information for different local models, the fuse of local information during multi-step model rollout may lead to the propagation of prediction errors between different local models, as discussed in the Section~\ref{sec:intro}.
To address this issue, we reformulate the model rollout process as the \textit{model MMDP} where the current joint policy is fixed to serve as the background environment and the local models are reversely regarded as decision-makers aiming at minimizing the global accumulative model error.

In the single-agent setting, some works have attempted to learn the model by treating the model rollout process as a sequential decision-making problem.
\citet{shang2019environment} propose an environment reconstruction method which models the influence of the hidden confounder on the environment by treating the platform, the user and the confounder as three agents interacting with each other. 
They focus on the offline setting (i.e., RL-based recommendation) and simultaneously train the model and the policy using a multi-agent imitation learning method.
\citet{xu2020error} treat the model as a dual agent and analyze the error bounds of the model. They propose to train the model using imitation learning methods. 
\citet{chen2022adversarial} also consider multi-step model errors, yet they mainly focus on handling counterfactual data queried by adversarial policies. 
Note that both~\citep{xu2020error} and~\citep{chen2022adversarial} focus solely on model learning in the single-agent setting and do not combine with the policy learning phase.

There are also works considering multi-step prediction loss in the single-agent setting~\citep{nagabandi2020deep,luo2018algorithmic}.
The essential difference between their multi-step loss and ours is that their loss is computed over the trajectories sampled from the environment dataset (collected by previous policies), while MAG minimizes the multi-step loss on the trajectories generated by active interactions between the local models as well as the current joint policy.
From the theoretical perspective, the model error term in Theorem 1 is defined by the expectation over the current joint policy and the current local models, thus computing the multi-step loss on the trajectories generated by these current policy and current models can better approximate the lower bound, which guarantees better policy improvement.

\section{Conclusion and Future Work}
In this work, we first study how the prediction errors of agent-wise local models affect the performance lower bound, which necessitates the considerations of the interactions between models during multi-step model rollout.
Based on this theoretical result, we reformulate the model rollout process as the \textit{model MMDP} by treating the local models as multi-step decision-makers and the current policies as the background environment.
We then propose a multi-agent model learning framework, i.e., MAG, to maximize the accumulative global ``model-reward'' defined in the \textit{model MMDP} by considering the interactions between local models.
We provide a practical implementation of MAG to optimize the above objective using the model predictive control.
Empirically, we show that MAG outperforms both model-based and model-free baselines on several challenging tasks in the StarCraft II benchmark, and the quantitative analysis of the model error further validates the effectiveness of our algorithmic design.
%We also carry out ablation studies to show how the choices of planning horizon and the number of random trajectories affect the overall performance.
For the future work, we plan to study the problem of learning local models in the competitive or mixed cooperative-competitive scenarios, which can be seen as learning in a \textit{Markov Game}.

\section{Acknowledgments}
We gratefully acknowledge the support from the National Natural Science Foundation of China (No.62076259), the Fundamental and Applicational Research Funds of Guangdong province (No.2023A1515012946), and the Fundamental Research Funds for the Central Universities-Sun Yat-sen University.

\bigskip
\bibliography{aaai23}

\clearpage
\noindent{\huge \textbf{Appendices}}
\bigskip
\appendix

\section{Mathematical Proofs}
\subsection{Proof of Theorem 1}
\begin{theorem}
\label{thm:upperbound}
The gap between the expected return of the model and the environment can be bounded as:
    \begin{align}
    \label{eq:upperbound}
\left| J(\boldsymbol{\pi})-J^{\hat{P}}(\boldsymbol{\pi})\right| \leq
\frac{R_{max}}{(1-\gamma)^2}\left(2\epsilon_{\boldsymbol{\pi}}+
(1-\gamma)\sum_{t=1}^\infty\gamma^t\epsilon_{m_t}\right),
\end{align}
where $\epsilon_{\boldsymbol{\pi}}:=\max_{\boldsymbol{o}} D_{TV}(\boldsymbol{\pi}_D(\cdot| \boldsymbol{o})\|\boldsymbol{\pi}(\cdot| \boldsymbol{o}))$ denotes the distribution shift of the joint policy between two consecutive iterations,
$\epsilon_{m_t}:=\mathbb{E}_{\boldsymbol{o}\sim \hat{P}_{t-1}(\boldsymbol{o};\boldsymbol{\pi})} \left[\max_{\boldsymbol{a}}\sqrt{2\sum_{i=1}^N\mathbb{E}_{\boldsymbol{o}'\sim \hat{P}(\cdot|\boldsymbol{o}, \boldsymbol{a})}\left[ \log\frac{\hat{P}^i(o^{i'}|\boldsymbol{o}, \boldsymbol{a})}{\sqrt[N]{P(\boldsymbol{o}'|\boldsymbol{o}, \boldsymbol{a})}}   \right]}\right] $ denotes the upper bound of the $i$-th model's error at timestep $t$ of the model rollout trajectory, 
$\hat{P}_{t-1}(\boldsymbol{o};\boldsymbol{\pi})$ denotes the distribution of joint observation at $t-1$ under $\hat{P}$ and $\boldsymbol{\pi}$, 
and $R_{max}:=\max_{s, \boldsymbol{a}} R(s, \boldsymbol{a})$.
%\left(\prod_{i}^N\hat{P}^i(\cdot| \boldsymbol{o}, \boldsymbol{a})\|\prod_{i}^NP^i(\cdot|\boldsymbol{o}, \boldsymbol{a})\right)
\end{theorem}

\begin{proof}
Note that in Section 2 of the main paper, we have assumed that $(\boldsymbol{o}, \boldsymbol{a})$ is sufficient to obtain $R$ without loss of generality, thus the reward function $R$ is written as $R(\boldsymbol{o}, \boldsymbol{a})$ in this proof.
To derive the relation of performance gap w.r.t. the policy shift during one training iteration, we introduce the data-collecting policy by adding and subtracting $J(\boldsymbol{\pi}_D)$, then we have:
\begin{align}
\label{eq:JMJP}
    \left| J(\boldsymbol{\pi})-J^{\hat{P}}(\boldsymbol{\pi})\right| \leq \left| J(\boldsymbol{\pi})-J(\boldsymbol{\pi}_D)\right| + \left| J(\boldsymbol{\pi}_D)-J^{\hat{P}}(\boldsymbol{\pi})\right|.
\end{align}
Applying Lemma~\ref{lemma:J1MJ2} on the two terms of the right hand side of Eq.~(\ref{eq:JMJP}), we have:
\begin{align}
    \left| J(\boldsymbol{\pi})-J^{\hat{P}}(\boldsymbol{\pi})\right|\leq 
    \frac{R_{max}}{(1-\gamma)^2}\left(2\epsilon_{\boldsymbol{\pi}}+
(1-\gamma)\sum_{t=1}^\infty\gamma^t\epsilon_{m_t}\right).
\end{align}

\end{proof}
\begin{lemma}
\label{lemma:J1MJ2}
Denoting $ J_1(\boldsymbol{\pi}_1) $ as the return of $\boldsymbol{\pi}_1$ under $p_1(\boldsymbol{o}'| \boldsymbol{o}, \boldsymbol{a})$ where $p_1(\boldsymbol{o}'| \boldsymbol{o}, \boldsymbol{a})=\prod_{i=1}^N p_1^i(o^{i'}| \boldsymbol{o}, \boldsymbol{a})$, $ J_2(\boldsymbol{\pi}_2) $ as the return of $\boldsymbol{\pi}_2$ under $p_2(\boldsymbol{o}'| \boldsymbol{o}, \boldsymbol{a})$, $\epsilon_{\boldsymbol{\pi}}:=\max_{\boldsymbol{o}} D_{TV}(\boldsymbol{\pi}_1(\cdot| \boldsymbol{o})\|\boldsymbol{\pi}_2(\cdot| \boldsymbol{o}))$, 
$\epsilon_{m_t}:=\mathbb{E}_{\boldsymbol{o}\sim p_{1, t-1}(\boldsymbol{o};\boldsymbol{\pi})} \left[\max_{\boldsymbol{a}}\sqrt{2\sum_{i=1}^N\mathbb{E}_{\boldsymbol{o}'\sim p_1(\cdot|\boldsymbol{o}, \boldsymbol{a})}\left[ \log\frac{p_{1}^i(o^{i'}|\boldsymbol{o}, \boldsymbol{a})}{\sqrt[N]{p_2(\boldsymbol{o}'|\boldsymbol{o}, \boldsymbol{a})}}   \right]}\right] $, 
and $p_{1,t-1}(\boldsymbol{o};\boldsymbol{\pi})$ as the distribution of joint observation at $t-1$ under $p_1$ and $\boldsymbol{\pi}$, we have:
\begin{align}
    \left| J_1(\boldsymbol{\pi}_1)-J_2(\boldsymbol{\pi}_2)\right|\leq \frac{R_{max}}{(1-\gamma)^2}\left(\epsilon_{\boldsymbol{\pi}}+
(1-\gamma)\sum_{t=1}^\infty\gamma^t\epsilon_{m_t}\right).
\end{align}
\end{lemma}
\begin{proof}
\begin{align}
\label{eq:J1MJ2}
    \left| J_1(\boldsymbol{\pi}_1)-J_2(\boldsymbol{\pi}_2)\right| &= \left|\sum_{t}\gamma^t \sum_{\boldsymbol{o}, \boldsymbol{a}}\left(p_{1, t}(\boldsymbol{o}, \boldsymbol{a})-p_{2, t}(\boldsymbol{o}, \boldsymbol{a})\right)R(\boldsymbol{o}, \boldsymbol{a}) \right| \nonumber\\
    &\leq \sum_t\sum_{\boldsymbol{o}, \boldsymbol{a}}\gamma^t R_{max}\left|p_{1, t}(\boldsymbol{o}, \boldsymbol{a})-p_{2, t}(\boldsymbol{o}, \boldsymbol{a}) \right|\nonumber\\
    &=2R_{max}\sum_t\gamma^t D_{TV}\left[p_{1, t}(\boldsymbol{o}, \boldsymbol{a})\|p_{2, t}(\boldsymbol{o}, \boldsymbol{a}) \right].
\end{align}
According to Lemma~\ref{lemma:1}, we have:
\begin{align}
\label{eq:DTVpOA}
    D_{TV}\left[p_{1, t}(\boldsymbol{o}, \boldsymbol{a})\|p_{2, t}(\boldsymbol{o}, \boldsymbol{a}) \right] \leq &D_{TV}\left[p_{1, t}(\boldsymbol{o})\|p_{2, t}(\boldsymbol{o}) \right] +\nonumber \\
    &\max_{\boldsymbol{o}}D_{TV}\left[\boldsymbol{\pi}_1(\cdot|\boldsymbol{o})\|\boldsymbol{\pi}_2(\cdot|\boldsymbol{o}) \right].
\end{align}
According to Lemma~\ref{lemma:2}, we can bound the term $D_{TV}\left[p_1^t(\boldsymbol{o})\|p_{2, t}(\boldsymbol{o}) \right]$as follows:
\begin{align}
\label{eq:obsDTV}
    &D_{TV}\left[p_{1, t}(\boldsymbol{o})\|p_{2, t}(\boldsymbol{o}) \right]\nonumber\\  \leq&\sum_t \mathbb{E}_{\boldsymbol{o}\sim\hat{P}_{t-1}(\boldsymbol{o};\boldsymbol{\pi}_1)}D_{TV}\left[p_1(\cdot|\boldsymbol{o})\|p_2(\cdot|\boldsymbol{o}) \right]\nonumber \\
    =&\sum_t \mathbb{E}_{\boldsymbol{o}\sim\hat{P}_{t-1}(\boldsymbol{o};\boldsymbol{\pi}_1)} \frac{1}{2} \sum_{\boldsymbol{o}'}|p_1(\boldsymbol{o}'|\boldsymbol{o})-p_1(\boldsymbol{o}'|\boldsymbol{o})|\nonumber \\
    =&\sum_t \mathbb{E}_{\boldsymbol{o}\sim\hat{P}_{t-1}(\boldsymbol{o};\boldsymbol{\pi}_1)} \frac{1}{2} \sum_{\boldsymbol{o}'}\left|\sum_{\boldsymbol{a}}p_1(\boldsymbol{o}', \boldsymbol{a}|\boldsymbol{o})-p_1(\boldsymbol{o}', \boldsymbol{a}|\boldsymbol{o})\right|\nonumber \\
    \leq&\sum_t \mathbb{E}_{\boldsymbol{o}\sim\hat{P}_{t-1}(\boldsymbol{o};\boldsymbol{\pi}_1)} \frac{1}{2} \sum_{\boldsymbol{o}',\boldsymbol{a}}\left|p_1(\boldsymbol{o}', \boldsymbol{a}|\boldsymbol{o})-p_1(\boldsymbol{o}', \boldsymbol{a}|\boldsymbol{o})\right|\nonumber \\
    =&\sum_t \mathbb{E}_{\boldsymbol{o}\sim\hat{P}_{t-1}(\boldsymbol{o};\boldsymbol{\pi}_1)} D_{TV}\left[p_1(\boldsymbol{o}', \boldsymbol{a}|\boldsymbol{o})\|p_1(\boldsymbol{o}', \boldsymbol{a}|\boldsymbol{o})\right] \nonumber\\
    \leq&\sum_t\mathbb{E}_{\boldsymbol{o}\sim\hat{P}_{t-1}(\boldsymbol{o};\boldsymbol{\pi}_1)}\Big[D_{TV}\left[\boldsymbol{\pi}_1(\cdot|\boldsymbol{o})\|\boldsymbol{\pi}_2(\cdot|\boldsymbol{o}) \right] + \nonumber\\
    &\mathbb{E}_{\boldsymbol{a}\sim\boldsymbol{\pi}_1}D_{TV}\left[p_1(\cdot|\boldsymbol{o}, \boldsymbol{a})\|p_2(\cdot|\boldsymbol{o}, \boldsymbol{a}) \right] \Big] \nonumber\\
    \leq&t\epsilon_{\boldsymbol{\pi}}+\sum_t\mathbb{E}_{\boldsymbol{o}\sim\hat{P}_{t-1}(\boldsymbol{o};\boldsymbol{\pi}_1)}\max_{\boldsymbol{a}}D_{TV}\left[p_1(\cdot|\boldsymbol{o}, \boldsymbol{a})\|p_2(\cdot|\boldsymbol{o}, \boldsymbol{a}) \right].
\end{align}
% using the independent assumption mentioned in Section 2, i.e., $\hat{P}(\boldsymbol{o}'|\boldsymbol{o}, \boldsymbol{a})= \prod_{i=1}^N \hat{P}^i(o^{i'}|\boldsymbol{o}, \boldsymbol{a})$
Applying the Pinsker's inequality to the term $D_{TV}\left[p_1(\cdot|\boldsymbol{o}, \boldsymbol{a})\|p_2(\cdot|\boldsymbol{o}, \boldsymbol{a}) \right]$ and decomposing $p_1$, we have:
\begin{align}
\label{eq:pinsker}
    &D_{TV}\left[p_1(\cdot|\boldsymbol{o}, \boldsymbol{a})\|p_2(\cdot|\boldsymbol{o}, \boldsymbol{a}) \right]\nonumber\\ \leq& \sqrt{\frac{1}{2} D_{KL}\left[p_1(\cdot|\boldsymbol{o}, \boldsymbol{a})\|p_2(\cdot|\boldsymbol{o}, \boldsymbol{a}) \right]}\nonumber\\
    =&\sqrt{\frac{1}{2}\sum_{i=1}^N \mathbb{E}_{\boldsymbol{o}'\sim p_1(\cdot|\boldsymbol{o}, \boldsymbol{a})}\left[ \log\frac{p_1^i(o^{i'}|\boldsymbol{o}, \boldsymbol{a})}{\sqrt[N]{p_2(\boldsymbol{o}'|\boldsymbol{o}, \boldsymbol{a})}} \right] }.
\end{align}
Plugging Eq.~(\ref{eq:pinsker}) in Eq.~(\ref{eq:obsDTV}), we have:
\begin{align}
\label{eq:DTVpO}
    D_{TV}\left[p_{1,t}(\boldsymbol{o})\|p_{2, t}(\boldsymbol{o}) \right]\leq t\epsilon_{\boldsymbol{\pi}}+\frac{1}{2}\sum_{t'=0}^t\epsilon_{m_{t'}}.
\end{align}
Plugging Eq.~(\ref{eq:DTVpO}) in Eq.~(\ref{eq:DTVpOA}), we have:
\begin{align}
\label{eq:DTVpOAResult}
    D_{TV}\left[p_{1, t}(\boldsymbol{o}, \boldsymbol{a})\|p_{2, t}(\boldsymbol{o}, \boldsymbol{a}) \right] \leq
    (t+1)\epsilon_{\boldsymbol{\pi}}+\frac{1}{2}\sum_{t'=0}^t\epsilon_{m_{t'}}.
\end{align}
Further plug Eq.~(\ref{eq:DTVpOAResult}) in Eq.~(\ref{eq:J1MJ2}), then we have:
\begin{align}
    \left| J_1(\boldsymbol{\pi}_1)-J_2(\boldsymbol{\pi}_2)\right|\leq \frac{R_{max}}{(1-\gamma)^2}\left(\epsilon_{\boldsymbol{\pi}}+
(1-\gamma)\sum_{t=1}^\infty\gamma^t\epsilon_{m_t}\right).
\end{align}
\end{proof}

\begin{lemma}
\label{lemma:1}
Suppose that $p_1(x,y)=p_1(x)p_1(y|x)$ and $p_2(x,y)=p_2(x)p_2(y|x)$, then we have:
\begin{align}
    D_{T V}\left(p_{1}(x, y) \| p_{2}(x, y)\right) \leq &D_{T V}\left(p_{1}(x) \| p_{2}(x)\right)+\\&\max _{x} D_{T V}\left(p_{1}(y \mid x) \| p_{2}(y \mid x)\right).
\end{align}
Alternatively, we have a tighter bound in terms of the expected total variation divergence:
\begin{align}
    D_{T V}\left(p_{1}(x, y) \| p_{2}(x, y)\right) \leq &D_{T V}\left(p_{1}(x) \| p_{2}(x)\right)+\\&\mathbb{E}_{x\sim p_1} D_{T V}\left(p_{1}(y \mid x) \| p_{2}(y \mid x)\right).
\end{align}

% \begin{align}
% &\quad\ D_{T V}\left(p_{1}(x, y) \| p_{2}(x, y)\right) \\ &=\frac{1}{2} \sum_{x, y}\left|p_{1}(x, y)-p_{2}(x, y)\right| \\
% &=\frac{1}{2} \sum_{x, y}\left|p_{1}(x) p_{1}(y \mid x)-p_{2}(x) p_{2}(y \mid x)\right| \\
% &=\frac{1}{2} \sum_{x, y}\left|p_{1}(x) p_{1}(y \mid x)-p_{1}(x) p_{2}(y \mid x)+\left(p_{1}(x)-p_{2}(x)\right) p_{2}(y \mid x)\right| \\
% & \leq \frac{1}{2} \sum_{x, y} p_{1}(x)\left|p_{1}(y \mid x)-p_{2}(y \mid x)\right|+\left|p_{1}(x)-p_{2}(x)\right| p_{2}(y \mid x) \\
% &=\frac{1}{2} \sum_{x, y} p_{1}(x)\left|p_{1}(y \mid x)-p_{2}(y \mid x)\right|+\frac{1}{2} \sum_{x}\left|p_{1}(x)-p_{2}(x)\right| \\
% &=E_{x \sim p_{1}}\left[D_{T V}\left(p_{1}(y \mid x) \| p_{2}(y \mid x)\right)\right]+D_{T V}\left(p_{1}(x) \| p_{2}(x)\right) \\
% & \leq \max _{x} D_{T V}\left(p_{1}(y \mid x) \| p_{2}(y \mid x)\right)+D_{T V}\left(p_{1}(x) \| p_{2}(x)\right)
% \end{align}
\end{lemma}
Lemma~\ref{lemma:1} is proved in the MBPO paper~\citep{janner2019trust} (Appendix B.1), so we only provide the result here.

\begin{lemma}
\label{lemma:2}
\begin{align}
    &D_{TV}\left[p_{1, t}(\boldsymbol{o})\|p_{2, t}(\boldsymbol{o}) \right]\nonumber\\  \leq&\sum_t \mathbb{E}_{\boldsymbol{o}\sim\hat{P}_{t-1}(\boldsymbol{o};\boldsymbol{\pi}_1)}D_{TV}\left[p_1(\cdot|\boldsymbol{o})\|p_2(\cdot|\boldsymbol{o}) \right].
\end{align}
\end{lemma}

Lemma~\ref{lemma:2} is proved in the MBPO paper~\citep{janner2019trust} (Appendix B.2), so we only provide the result here.

\subsection{The Tightness of the Bound in Theorem~1}
In Section~3, we mention that the bound in Theorem~1 is tighter than the bound induced by scaling up $\epsilon_{m_t}$ to its maximum over timesteps, i.e., $\max_t\epsilon_{m_t}$.
Now we give a formal proof of the above claim.

\begin{proof}
If $\epsilon_{m_t}$ is scaled up to its maximum over timesteps in Lemma~\ref{lemma:J1MJ2}, then Eq.~(\ref{eq:obsDTV}) becomes:
\begin{align}
D_{TV}\left[p_{1, t}(\boldsymbol{o})\|p_{2, t}(\boldsymbol{o}) \right]\leq
    t\left(\epsilon_{\boldsymbol{\pi}}+\delta\right),
\end{align}
where $\delta=\max_{\boldsymbol{a},t}\mathbb{E}_{\boldsymbol{o}\sim \hat{P}_{t-1}(\boldsymbol{o};\boldsymbol{\pi})}D_{TV}\left[p_1(\cdot|\boldsymbol{o}, \boldsymbol{a})\|p_2(\cdot|\boldsymbol{o}, \boldsymbol{a}) \right]$.
By repeatedly plugging equations like the proof of Lemma~\ref{lemma:J1MJ2}, the induced bound becomes:
\begin{align}
\label{eq:tightness}
    \left| J(\boldsymbol{\pi})-J^{\hat{P}}(\boldsymbol{\pi})\right|\leq 
    \frac{R_{max}}{(1-\gamma)^2}\left(2\epsilon_{\boldsymbol{\pi}}+
\gamma\delta\right).
\end{align}
Subtracting Eq.~(\ref{eq:upperbound}) by Eq.~(\ref{eq:tightness}), we have:
\begin{align}
    &\frac{R_{max}}{(1-\gamma)^2}\left(2\epsilon_{\boldsymbol{\pi}}+
(1-\gamma)\sum_{t=1}^\infty\gamma^t\epsilon_{m_t}\right) - \frac{R_{max}}{(1-\gamma)^2}\left(2\epsilon_{\boldsymbol{\pi}}+
\gamma\delta\right)\nonumber  \\
=&\frac{R_{max}}{(1-\gamma)^2}\left(\gamma\delta - (1-\gamma)\sum_{t=1}^\infty\gamma^t\epsilon_{m_t} \right)\nonumber \\
\geq&\left(\gamma\delta - (1-\gamma)\sum_{t=1}^\infty\gamma^t\delta \right)=0, 
\end{align}
In other words, the bound in Theorem~1 is tighter than the bound in Eq.~(\ref{eq:tightness}), thus proving the claim.

\end{proof}

\section{Implementation Details}
The implementation of MAG is overall built on MAMBA~\citep{egorov2022scalable}, and the network structures as well as the hyperparameter settings of the underlying MAMBA component remain completely the same as the original implementation of MAMBA~\footnote{https://github.com/jbr-ai-labs/mamba}, except that the learning rate of the model is set to 5e-4 in MAG.
The incremental network component of MAG w.r.t. MAMBA is a fully-connected network that approximates $R_m$, which has 4 layers and 256 neurons in each layer.
As discussed in Section~4, there are mainly two extra hyperparameters in MAG that need to be tuned, i.e., the planning horizon $H$ and the number of random shooting trajectories $L$.
Here we specify the settings of these two hyperparameters in Table~\ref{tab:hyper}.
\begin{table}[h]

    \centering
    \begin{tabular}{lll}
    \hline
        ~ & L & H \\ \hline
        2s\_vs\_1sc & 5 & 10 \\ 
        3s\_vs\_3z & 5 & 10 \\
        2s3z & 5 & 10 \\ 
        3s\_vs\_4z & 4 & 10 \\ 
        3s\_vs\_5z & 5 & 7 \\
        2c\_vs\_64zg & 5 & 7 \\ 
        corridor & 4 & 6 \\ 
        3s5z\_vs\_3s6z & 4 & 7 \\
    \hline
    \end{tabular}
    \captionof{table}{Settings of $L$ and $H$.}
    \label{tab:hyper}
\end{table}

\end{document}